\documentclass{article}

    \PassOptionsToPackage{numbers, compress}{natbib}
    \usepackage[preprint]{neurips_2025}

\usepackage[utf8]{inputenc} % allow utf-8 input
\usepackage[T1]{fontenc}    % use 8-bit T1 fonts
\usepackage{hyperref}       % hyperlinks
\usepackage{url}            % simple URL typesetting
\usepackage{booktabs}       % professional-quality tables
\usepackage{amsfonts}       % blackboard math symbols
\usepackage{nicefrac}       % compact symbols for 1/2, etc.
\usepackage{microtype}      % microtypography
\usepackage{xcolor}         % colors

\usepackage{amsmath, amssymb, mathtools, amsthm}
\theoremstyle{plain}
\newtheorem{theorem}{Theorem}

\usepackage{algorithm}
\usepackage[noend]{algpseudocode}
\floatstyle{ruled}
\restylefloat{algorithm}

\usepackage{graphicx}
\usepackage{subcaption}

\usepackage{listings}
\usepackage{xcolor}
\ifthenelse{1=1}{
\newcommand{\todo}[1]{{\color{red}{TODO: #1}}}
\newcommand{\saman}[1]{{\color{orange}{Saman: #1}}}
}{
\newcommand{\todo}[1]{}
\newcommand{\saman}[1]{}
}

\title{Eliminating Hallucination-Induced Errors in LLM Code Generation with Functional Clustering}

\author{
  Chaitanya Ravuri \\
  Massachusetts Institute of Technology \\
  \texttt{cravuri@mit.edu} \\
  \And
  Saman Amarasinghe \\
  Massachusetts Institute of Technology \\
  \texttt{samana@mit.edu} \\ 
}

\begin{document}

\maketitle

\begin{abstract}
  Modern code–generation LLMs can already solve a large fraction of programming problems, yet they still hallucinate subtle bugs that make their outputs unsafe for autonomous deployment.  We present \emph{functional clustering}, a black-box wrapper that eliminates nearly all hallucination-induced errors while providing a tunable confidence score. The wrapper samples many candidate programs, executes each on a self-generated test suite, and clusters candidates whose I/O behavior is identical; the empirical mass of the largest cluster serves as an exact confidence estimate. A single scalar threshold on this estimate lets users trade coverage for reliability with exponential guarantees. On \textsc{LiveCodeBench} our verifier preserves baseline pass@1 on solvable tasks yet slashes the error rate of returned answers from $\sim$65\% to 2\%, and drives it to 0\% at a conservative threshold while still answering 15.6\% of prompts. Manual audits show that the few residual mistakes stem from prompt misinterpretation, not random generation noise, narrowing future work to specification clarity. Because the method requires only sampling and sandbox execution, it applies unchanged to closed-source APIs and future models, offering a practical path toward dependable, autonomous code generation. Our code is available on Github (\href{https://github.com/20ChaituR/functional-clustering}{this https URL}).
\end{abstract}

\section{Introduction}
\label{sec:intro}

% \saman{One angle you can take is that each generation of LLMs have improve the average accuracy, but it is not making progress towards reducing the errors. (Need to define the expectation of programming.   They are expected to be correct all the time (However, the correctness is normally not measured by proving, but by exhaustively testing. According to Chat-GPT “Three nines” (99.9\%) for many consumer-facing web services “Four to five nines” (99.99–99.999\%) for financial, healthcare, or telecom systems “Six nines+” (99.9999\%+) only where even a few seconds of downtime is absolutely unacceptable (e.g. some real-time control systems)).   You can use the example that the latest model create the wrong program with much more strong conviction than the earlier model, to motivate complimentary techniques to reduce the errors.}

Large language models (LLMs) have begun to write code at a level once thought exclusive to experienced engineers: they pass university exams \cite{nezhadllmperformance}, solve competitive‑programming problems \cite{openai2025competitiveprogramminglargereasoning}, and even draft patches that compile in production repositories \cite{taomagis}. Despite this progress, practitioners hesitate to deploy them unsupervised \cite{zhao2025surveylargelanguagemodels}. A single off‑by‑one error or mistyped logical operator can crash an application or leak sensitive data, and state‑of‑the‑art models still hallucinate such bugs with disconcerting frequency \cite{Liu}.

Most recent work therefore augments LLMs with confidence scores \cite{xiong2024llmsexpressuncertaintyempirical}. Token‑level likelihoods \cite{ma2025estimatingllmuncertaintylogits}, calibrated logits \cite{desai-durrett-2020-calibration}, and semantic‑embedding clustering \cite{qiu2024semanticdensityuncertaintyquantification} each offer partial signals, yet all miss a fundamental aspect of code. Two snippets that differ syntactically--or even occupy completely separate regions of embedding space--may be functionally identical. At the same time, semantically similar code with just single-character edits (such changing a $<$ to $\le$) flips correctness without appreciably moving the embedding. Treating sequences or embeddings as the unit of analysis necessarily conflates these cases.

Our key observation is that software engineering already supplies an exact, model‑agnostic criterion for identity: behavior on test inputs. In practice, virtually every production program is accepted not because it is formally proved but because it passes a finite test suite.  We leverage the same idea at generation time. For each coding prompt the LLM produces a bag of candidate programs; the same model is then prompted to generate a diverse set of inputs (only inputs—no labeled outputs).  Executing every program on every input in a sandbox yields an output vector; candidates whose vectors match exactly are placed in the same functional equivalence class. The empirical probability that a random sample belongs to the largest class, denoted $\rho$, becomes an \emph{exact} confidence estimate.

This test‑based clustering confers three immediate advantages.  
First, it aligns the score with the real goal--correctness of behavior--rather than proxies such as token probability or embedding distance. If the model assigns \(30\%\) probability to two syntactically different but equivalent implementations, $\rho$ correctly reports \(60\%\) confidence, whereas logit‑based metrics would cap the score at \(30\%\).  
Second, the approach is entirely black‑box: it requires only sampling and sandbox execution, not internal logits, gradient access, or fine‑tuning. Consequently it wraps closed‑source APIs as easily as open‑weight models. Third, the statistic separates failure modes. Random‑chance hallucinations fall exponentially fast with sample size; the rare high‑confidence errors we observe stem exclusively from prompt misinterpretations. By isolating misunderstandings, the method turns reliability into a research problem that is both narrower and more tractable.

% \saman{We need to have a summary of what we are doing, what is the big insight, and a list of contributions (including the results) before we talk about related work. Contributions should be a bulleted list.}

\paragraph{Contributions}
\begin{enumerate}
  \item We frame uncertainty in code generation as density estimation over functional equivalence classes--programs that produce identical outputs on an automatically‑generated test suite. This shifts the problem from approximate semantic similarity to exact behavioral identity, giving a mathematically crisp confidence signal without learned similarity models or log‑probabilities.

  \item The method requires only (i) sampling multiple programs and (ii) running them in a sandbox. It attaches to any code LLM, including closed APIs, without fine‑tuning, auxiliary classifiers, or access to internal logits. It does not require any additional knowledge beyond the prompt, as it generates its own test data.

  \item By prompting the base model for inputs only (not labeled test cases), we create a task‑agnostic verifier that introduces no additional source of hallucination and zero human overhead.

  \item A universal threshold on dominant‑cluster mass converts the LLM into a selective coder that either returns a representative program or abstains, giving a very high confidence of producing a correct solution when responding.
\end{enumerate}

In short, functional clustering converts noisy LLM outputs into a confidence--based response, bringing dependable, autonomous code generation within reach. Empirically, wrapping a modern code LLM with our verifier preserves baseline pass@1 on solvable benchmark tasks while slashing the error rate of returned answers from roughly \(65\%\) to \(2\%\). Raising the acceptance threshold drives residual error to zero at the cost of additional abstentions, giving users an explicit accuracy–coverage knob.  Manual inspection confirms that remaining failures are all specification misunderstandings—no random hallucinations survive.

\section{Related Work}

Hallucinations in LLMs have attracted considerable attention, especially in precision-critical domains like code generation. Current research seeks to understand and classify these phenomena, mitigate their occurrence, and develop robust evaluation benchmarks. Below, we survey key directions in this area, highlighting gaps that our approach aims to address.

Hallucinations are outputs that read plausibly yet are semantically wrong. Huang et al. \cite{Huang} and Islam et al. \cite{Islam} organize natural-language hallucinations into factual, logical, and stylistic errors and catalog mitigation recipes across prompting, decoding and post-hoc verification.  Code introduces new failure modes: HalluCode \cite{Liu} shows that models mismap argument roles, misname identifiers and mismanage resources even when they compile, while CodeHalu \cite{Tian} adds an execution-based benchmark that separates mapping, naming, resource and logic slips, demonstrating that functional correctness and intent fidelity remain brittle despite high surface fluency.

Multiple methods attempt to avoid these errors. One proposed method is to let the LLM grade its own answer. Kadavath et al. \cite{Kadavath} prompt GPT-3 to output a verbal probability “P(True)” and find that the statement is better calibrated than raw token likelihoods, while Tian et al. \cite{TianCalibrate} attach a value head fine-tuned with RLHF to supply confidence scores; both approaches improve in-domain calibration but generalize poorly to unseen tasks according to the broad empirical audit of Xiong et al. \cite{xiong2024llmsexpressuncertaintyempirical}.

A second method calibrates the token distribution itself  Desai and Durrett \cite{desai-durrett-2020-calibration} show that simple temperature scaling tightens expected calibration error for BERT-style classifiers, and Ma et al. \cite{ma2025estimatingllmuncertaintylogits} extend the idea to text generators by analyzing logit dispersion. Because token-based scores regard every distinct sequence as a different outcome, they must apportion probability mass across all syntactic variants of the same behaviour. If two programs that return identical outputs appear with 0.3 likelihood each, the model reports only 0.3 confidence for either one instead of 0.6 for the underlying solution. Thresholds or calibration curves built on that deflated score will wrongly flag many correct answers as “low-confidence,” forcing users to accept lower coverage or loosen the threshold and let more real errors slip through.

Supervised confidence heads trade label cost for sharper selectivity. Lin et al. \cite{LinTeaching} fine-tune GPT-3 to express its uncertainty in words; Liu et al. \cite{LiuLitCab} train a lightweight linear adapter over the last hidden layer to predict a bias term that rescales logits. Though fine-tuned heads produce sharp confidence curves where they are trained, each new domain or model demands fresh annotated triples of \textit{prompt, candidate code, correctness}. Collecting those labels is costly, often impossible for proprietary APIs, and must be repeated whenever the base model is updated—so the approach cannot serve as a drop-in wrapper across tasks or providers.

Embedding-based clustering targets semantic rather than lexical similarity.  Farquhar et al. \cite{Farquhar} introduce semantic entropy: they embed multiple generations, cluster by cosine similarity, and measure entropy to flag uncertainty, while Kuhn and Gal \cite{KuhnSemanticUncertainty} and Qiu and Miikkulainen \cite{qiu2024semanticdensityuncertaintyquantification} refine the notion with invariance-aware kernels and density estimates.  For natural language these groupings align with meaning, but for code the mapping from embedding space to behavior is loose. In code, moving from \texttt{<} to \texttt{<=}, a one-token change that may leave the embedding almost unchanged, can flip every test outcome, while two implementations with totally different variable names and control flow may sit far apart in the vector space yet return identical results. Relying on that geometry therefore blurs the line between correct and buggy solutions.

A recent code-specific variant by Sharma and David \cite{sharma2025assessingcorrectnessllmbasedcode} executes symbolic traces to decide whether two programs match, then applies an entropy measure over trace clusters. Symbolic traces compare the exact sequence of states a program visits, so any change in control flow—a loop unrolled one extra time, a recursive call replaced by iteration—yields a different trace even if the outputs on all inputs are identical. As a result, probability mass that ought to accumulate on one correct solution is scattered across several trace “modes,” lowering the reported confidence for each and again driving overly cautious abstention thresholds.

Selective-answer frameworks build reliability guarantees on top of any base score. Abbasi Yadkori et al. \cite{Yadkori} adapt conformal prediction to language models, proving error-rate bounds when the model abstains below a threshold, and Ye et al. \cite{YeBenchmarkingUQ} benchmark twenty UQ metrics across tasks, confirming that abstention policies can dominate naive always-answer baselines if the underlying score is well-aligned with correctness.

Despite this progress no existing signal is simultaneously behavioral, exact and black-box. Our functional-clustering verifier fills the gap by treating “produces identical outputs on an automatically generated test suite” as the atomic equivalence relation, merging syntactically diverse but behaviorally identical implementations and concentrating probability mass where it belongs. Because it needs only sampling and sandbox execution, the method plugs directly into the abstention frameworks above and wraps either open-weight or proprietary models without retraining or logit access.

To verify that functional clustering delivers practical reliability, we test it on two widely adopted benchmarks that cover complementary slices of the code-generation landscape. HumanEval \cite{HumanEval} consists of 164 short Python functions, each equipped with hidden unit tests that probe corner cases; because every prompt is self-contained and free from training contamination, pass@k here isolates the model’s raw reasoning and synthesis skills and reveals how much accuracy is recovered once low-confidence answers are filtered out. LiveCodeBench \cite{LiveCodeBench} pushes the same models into longer, resource-bound problems executed in sandboxes under real-time and memory limits; tasks require reading from \texttt{stdin}, writing to \texttt{stdout}, and handling large inputs, so they approximate continuous-integration pipelines and competitive-programming judges. Success on HumanEval shows that the verifier preserves baseline coverage on micro-level logic, while gains on LiveCodeBench demonstrate that the same mechanism scales to production-style constraints where a single hidden bug can lead to an incorrect solution.

\section{Methodology}\label{sec:methodology}

This section lays out the full technical scaffold behind our verifier. We begin by formalizing the task and the notion of \emph{functional equivalence} (Section \ref{sec:prob-statement}). We then derive an exact—but usually intractable—confidence metric based on the probability mass of an equivalence class and show how it can be estimated from finite samples (Section \ref{sec:fcp}). Next, we replace undecidable equivalence with a test‑based proxy and analyze the resulting error bounds (Section \ref{sec:test-oracle}). Finally, we assemble these pieces into a practical inference algorithm and discuss its computational profile (Section \ref{sec:algorithm}).

\subsection{Problem statement}\label{sec:prob-statement}

Let $\mathcal X$ and $\mathcal Y$ denote the input and output domains of a programming task described by natural‑language prompt $\mathbf p$. Conditioned on $\mathbf p$, an auto-regressive LLM defines a predictive distribution $\mathbb P$ over syntactically valid programs $\phi\in\mathcal P$. For a concrete program $\phi\in\mathcal P$ its semantics is the
deterministic function $f_\phi:\mathcal X\to\mathcal Y$.

Two programs are \textit{functionally equivalent} when they return the same output on every possible input:
\begin{equation}\label{eq:equiv}
\phi_1 \equiv \phi_2
\;\Longleftrightarrow\;
\forall x\in\mathcal X,\;
f_{\phi_1}(x)=f_{\phi_2}(x).
\end{equation}
Throughout the paper we assume that \emph{each task has a single correct output for any given input}. That is, if two programs differ in their behavior on even one test case they cannot both be fully correct. Some benchmark problems violate this assumption--for example, tasks that accept any permutation, any tie-breaking order, or any string that matches a regular expression. We treat such tasks as out of scope, but note that a similar method as ours could be used to determine equivalency in such problems.

Our aim is to (i) estimate the probability mass that $\mathbb P$ assigns to the equivalence class of the program we ultimately return and (ii) output that program only when the mass exceeds a user‑chosen threshold $\tau\in(0,1]$; otherwise we abstain. Equivalence classes are defined by Eq. \ref{eq:equiv}. The intuition is that genuine solutions tend to coalesce: there are many syntactic ways to implement the same correct algorithm, so correct programs accumulate probability in a shared equivalence class. By contrast, hallucinated errors are essentially random and therefore split their probability mass across many small, disjoint classes. A large observed class is thus strong evidence of correctness, whereas a small class signals either rarity or error.

Because universal equivalence is undecidable, we approximate it with agreement on a finite, automatically generated test set and show that the resulting estimator inherits exponential reliability guarantees.

\subsection{Empirical confidence}\label{sec:fcp}

The exact confidence of $\hat\phi$ is
\begin{equation}\label{eq:conf}
C(\hat\phi)\;=\;\Pr_{\Phi\sim\mathbb P} \bigl[\Phi\equiv\hat\phi\bigr],
\end{equation}
i.e. the total probability mass of its equivalence class. Calculating $C$ according to Eq. \ref{eq:conf} is typically intractable, so we resort to Monte‑Carlo:
\begin{equation}\label{eq:monte-carlo}
\hat C_n(\hat\phi)
= \frac1n\sum_{i=1}^{n}\mathbf 1\!\bigl[\Phi_i\equiv\hat\phi\bigr],
\qquad
\Phi_i\stackrel{\text{i.i.d.}}{\sim}\mathbb P.
\end{equation}
The indicator in Eq. \ref{eq:monte-carlo} is Bernoulli with mean $C(\hat\phi)$, hence $\hat C_n$ is unbiased and its variance scales as $C(1-C)/n$.

\subsection{Abstention rule and its reliability}
\label{sec:abstention}

We \textsc{answer} when $\hat C_n(\hat\phi)\ge\tau$ and \textsc{abstain} otherwise. The probability of a (potentially harmful) false acceptance is
\begin{equation}\label{eq:prob-harmful}
\Pr[\hat C_n\ge\tau\mid C]
\;=
\sum_{k=\lceil n\tau\rceil}^{n}\binom{n}{k}C^{k}(1-C)^{n-k}.
\end{equation}
When $C<\tau$, the tail in Eq. \ref{eq:prob-harmful} can be bounded with exponential tightness:
\begin{equation}\label{eq:chernoff}
\Pr[\hat C_n\ge\tau\mid C]
\;\le\;
\exp\!\bigl[-n\,D_{\mathrm{KL}}(\tau\|C)\bigr],
\end{equation}
where $D_{\mathrm{KL}}$ is the binary KL divergence. The quantity $D_{\mathrm{KL}}(\tau\|C)$ measures the per‑sample information cost of confounding a Bernoulli distribution of mean $C$
with one of mean $\tau$. Whenever $\tau\neq C$ this cost is strictly positive, so the error probability decays exponentially in $n$. Doubling the sample size approximately squares the bound, providing a simple knob for practitioners to dial target reliability. The proof of this bound is given in Appendix \ref{app:chernoff-proof}.

In our experiments we take $n=100$. With a moderate gap $\tau-C=0.2$ the divergence is $D_{\mathrm{KL}}(\tau\|C)\approx 0.14$, so $nD_{\mathrm{KL}}\approx 14$ and Eq. \ref{eq:chernoff} upper‑bounds the false‑accept probability by $e^{-14}\approx10^{-6}$. This level of assurance is sufficient for practical deployment; larger $n$ or a wider gap further suppresses risk at the cost of additional LLM queries. Thus, the abstention rule converts the intrinsic randomness of LLM sampling into a rigorously quantified, exponentially small chance of returning a low-confidence solution.

\subsection{Practical equivalence via testing}\label{sec:test-oracle}

Exact functional equivalence is coNP‑hard in general, and undecidable when programs may not halt.  Mirroring real-world software developmental practice, we use a test‑based oracle. Let $\mathcal S=\{x_1,\dots,x_m\}$ be a set of test inputs drawn from a task‑dependent distribution $\mathcal D$. We declare
\begin{equation}\label{eq:test-equiv}
\phi_1\equiv_{\mathcal S}\phi_2
\quad\Longleftrightarrow\quad
f_{\phi_1}(x_j)=f_{\phi_2}(x_j)\;\text{ for all }j.
\end{equation}
Suppose two programs disagree on a measurable subset
\(\mathcal B\subseteq\mathcal X\) with \(\mathcal D(\mathcal B)=\delta>0\). The oracle defined in Eq. \ref{eq:test-equiv} only fails to expose this difference precisely when none of the test inputs $\mathcal S$ are in $\mathcal B$. The probability of that occurring is
\begin{equation}\label{eq:test-prob}
\Pr[\phi_1\equiv_{\mathcal S}\phi_2]
\;=\;
(1-\delta)^{m}.
\end{equation}
Thus, with only hundreds of random tests, the probability in Eq. \ref{eq:test-prob} is exponentially unlikely unless the behavioral divergence itself is vanishingly small. In practice most real defects (off‑by‑one errors, edge‑case branches, etc.) affect a sizable slice of the input space, so such a guarantee suffices. The outer Chernoff exponent in $n$ from Eq. \ref{eq:chernoff} and the inner detection exponent in $m$ from Eq. \ref{eq:test-prob} compound multiplicatively: increasing either budget tightens overall guarantees without hidden interactions. Consequently our procedure inherits the best of both worlds—\emph{statistical rigor} from concentration inequalities and \emph{engineering practicality} from black‑box testing—while remaining compatible with modern LLM pipelines that already generate large candidate batches for re-ranking or self‑consistency.

\subsection{Function clustering routine}\label{sec:algorithm}

Algorithm \ref{alg:fc} summarizes the complete pipeline executed at inference time.  The routine is parameter‑free except for three user‑visible knobs: the number of candidate programs $n$, the number of automatically generated test inputs $m$, and the acceptance threshold $\tau$.

\begin{algorithm}[t]
\caption{Functional Clustering}\label{alg:fc}
\begin{algorithmic}[1]
\Require task description $\mathbf p$ (string); sample budgets $n$ (programs) and $m$ (inputs)\; threshold $\tau\in(0,1]$
\Ensure \textsc{abstain} or a high‑confidence program $\hat\varphi$

\Statex\textbf{Step 1: Candidate program sampling}
\For{$i \gets 1$ \textbf{to} $n$}
    \State $\varphi_i \gets \textsc{LLMGenerateProgram}(\mathbf p)$
\EndFor

\Statex\textbf{Step 2: Test‑input generation}
\For{$j \gets 1$ \textbf{to} $m$}
    \State $x_j \gets \textsc{LLMGenerateInput}(\mathbf p)$
\EndFor

\Statex\textbf{Step 3: Behavioural execution}
\For{$i \gets 1$ \textbf{to} $n$}
    \For{$j \gets 1$ \textbf{to} $m$}
        \State $o_{ij} \gets \textsc{SandboxRun}(\varphi_i, x_j)$
    \EndFor
    \State $\mathbf o_i \gets (o_{i1},\dots,o_{im})$ \Comment{output vector}
\EndFor

\Statex\textbf{Step 4: Equivalence clustering}
\State Partition $\{\varphi_i\}_{i=1}^n$ into classes $\mathcal C_1,\dots,\mathcal C_K$ where $\mathbf o_a=\mathbf o_b$ iff $a,b\in\mathcal C_k$
\State $s_{\max} \gets \max_k |\mathcal C_k|$

\Statex\textbf{Step 5: Decision}
\If{$s_{\max} \ge \lceil \tau n\rceil$}
    \State \textbf{return} any $\varphi \in \operatorname*{arg\,max}_k |\mathcal C_k|$ \Comment{high‑confidence answer}
\Else
    \State \textbf{return} \textsc{abstain}
\EndIf
\end{algorithmic}
\end{algorithm}

The verifier proceeds in three steps.  An LLM is first queried $n$ times for candidate programs $\varphi_{1:n}$ and $m$ times for test inputs $x_{1:m}$. We then execute every program on every input in a sandbox, recording the output vectors $\mathbf{o}_i=(o_{i1},\dots,o_{im})$. Candidates whose vectors match exactly are grouped; if the largest group contains at least $\tau n$ elements we return any representative (all are functionally equivalent on the tests), otherwise we output \textsc{abstain}.  Exact I/O equality replaces embedding or logit heuristics with a behavioral notion of equivalence.

\paragraph{Computational cost.}  The procedure issues $n+m$ LLM calls, so $T_{\mathrm{LLM}}=\mathcal{O}(n+m)$ dominates wall‑clock time. Local sandbox evaluation performs $nm$ runs, $T_{\mathrm{exec}}=\mathcal{O}(nm)$, but those runs are fast relative to the LLM calls. Memory is linear in $n+m$ because each vector component is constant‑size.

\section{Experiments}
\label{sec:experiments}

We evaluate our verifier on \textsc{HumanEval} and \textsc{LiveCodeBench} using two code LLMs: GPT-4o \cite{openai2024gpt4ocard} and Claude-3-Haiku \cite{Anthropic2024Claude3}. For each task we sample 50 candidate programs from each model (100 total) with chain-of-thought plus code prompting; all prompts are reproduced verbatim in Appendix \ref{app:prompts}. The self-generated test suites are produced by GPT-4o alone. We then cluster the 100 programs by exact I/O behavior and apply the thresholding rules introduced in Section \ref{sec:methodology}. Unless stated otherwise, metrics on \textsc{HumanEval} use only GPT-4o generations, while \textsc{LiveCodeBench} results use the full 100-sample pool.

\subsection{HumanEval}

\begin{figure}
  \centering
  \begin{subfigure}{\textwidth}
    \centering
    \includegraphics[width=13cm]{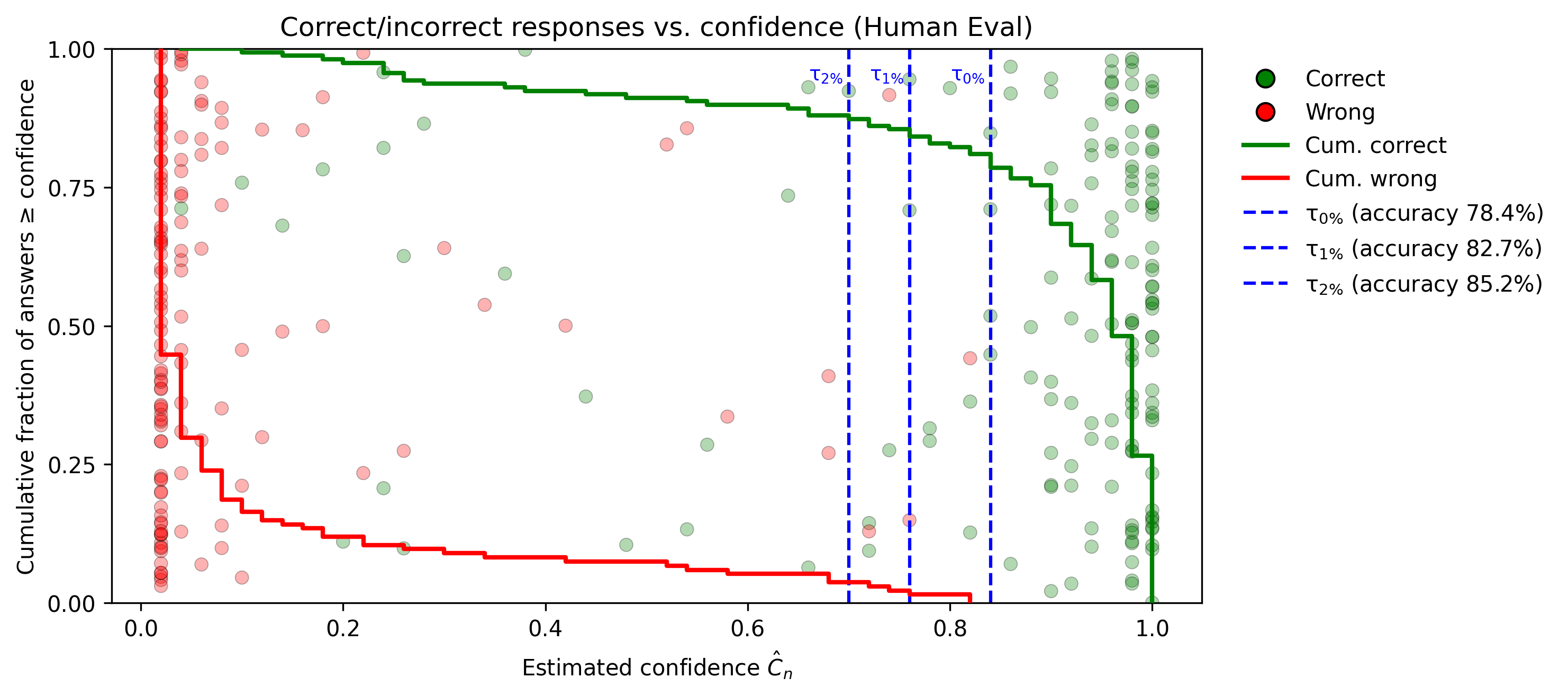}
    \caption{Correctness versus estimated confidence on \textsc{HumanEval}.}
    \label{fig:scatter-humaneval}
  \end{subfigure}
  
  \vspace{0.5cm} % some vertical space between plots

  % Second plot
  \begin{subfigure}{\textwidth}
    \centering
    \includegraphics[width=13cm]{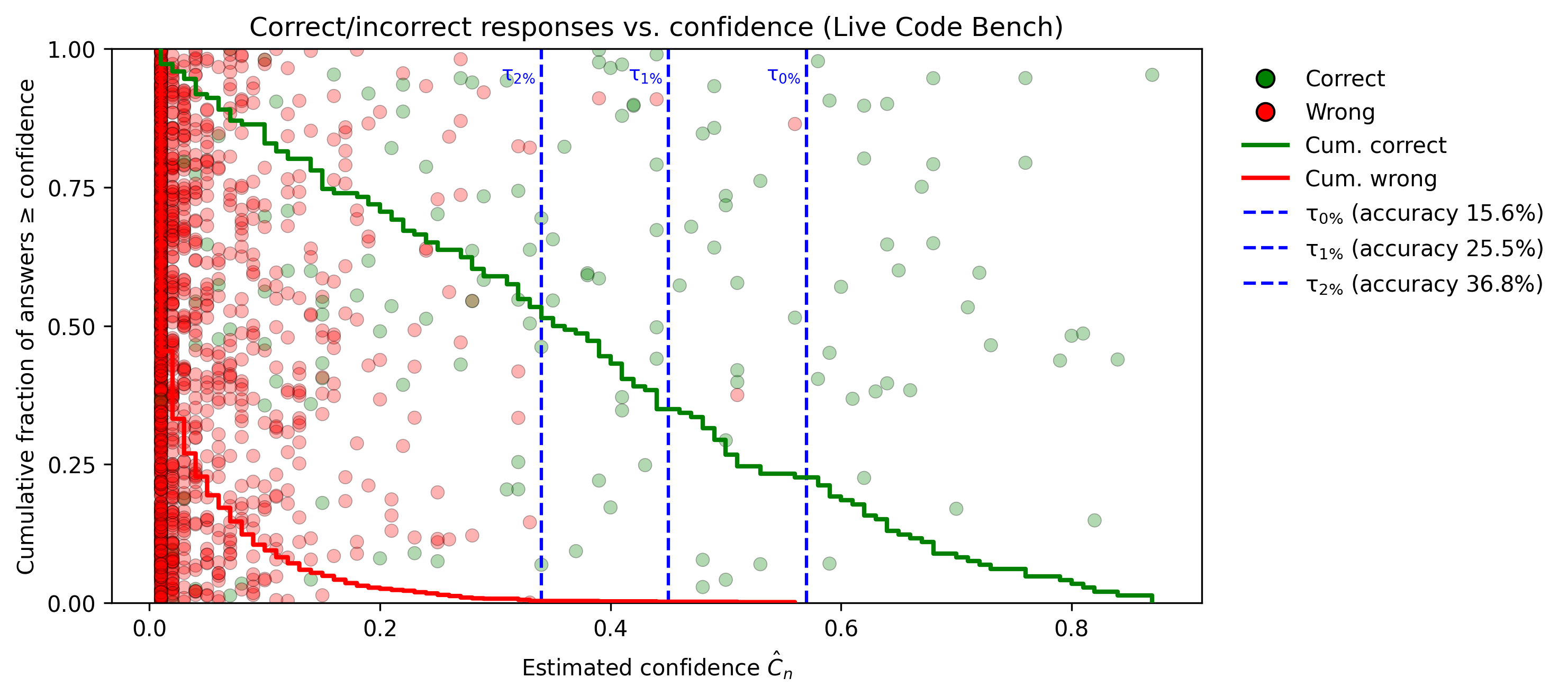}
    \caption{Correctness versus estimated confidence on \textsc{LiveCodeBench}.}
    \label{fig:scatter-lcb}
  \end{subfigure}
  
  \caption{\textbf{Correctness versus estimated confidence across datasets.}  Colors encode correctness; Each point is a response that is either correct or wrong, with its $x$ position denoting the model's confidence in the response, and its $y$ position meaningless. Dashed lines mark confidence thresholds, each guaranteeing an empirical error rate of at most the indicated percentage for returned answers. The Cum. Wrong / Cum. Correct step curves plot the cumulative percentage of incorrect and correct programs whose confidence lies above each point on the $x$-axis.}
  \label{fig:scatter}
\end{figure}

Figure~\ref{fig:scatter-humaneval} plots correctness versus estimated confidence $\hat C_n$ on \textsc{HumanEval}. The two step-shaped traces labeled \textit{Cum. Wrong} and \textit{Cum. Correct} summarize the scatter: at every $x$-axis value they report, respectively, the percentage of incorrect and of correct submissions whose confidence is greater than that value. Reading the curves from right to left therefore shows how residual error (red curve) and retained coverage (green curve) evolve as one lowers a single acceptance threshold. The baseline pass@1 of GPT-4o is $84.9$\%.  With functional clustering we achieve $85.2$\% accuracy at the $\tau_{2\%}$ operating point. Only four tasks lie above the threshold with incorrect code; manual inspection shows that every one of them stems from a prompt misunderstanding rather than a hallucinated bug.

To probe these outliers we rewrote each problematic specification, adding a single clarifying sentence. After the rewrites the model produces high-confidence correct solutions for all four tasks, confirming that the verifier pinpoints specification ambiguity rather than generation noise. Full rewrites are included in Appendix \ref{app:rewrites}.

% \begin{figure}[t]
%   \centering
%   \begin{subfigure}{0.46\textwidth}
% \begin{lstlisting}
% def order_by_points(nums):
%     """
%     Write a function which sorts the 
%     given list of integers in ascending 
%     order according to the sum of their 
%     digits.
%     """
% \end{lstlisting}
%     \caption{Excerpt from the original specification}
%   \end{subfigure}\hfill
%   \begin{subfigure}{0.46\textwidth}
% \begin{lstlisting}
% def order_by_points(nums):
%     """
%     Write a function which sorts the 
%     given list of integers in ascending 
%     order according to the sum of their 
%     digits.

%     If an integer is negative, its 
%     first digit is treated as negative 
%     and the remaining digits as 
%     positive.
%     """
% \end{lstlisting}
%     \caption{Clarified excerpt}
%   \end{subfigure}
%   \caption{\textbf{Manual rewrite of a misinterpreted \textsc{HumanEval} prompt.} Only the relevant portion of the docstring is shown. Adding the clarification resolves the model’s ambiguity and moves the solution into the high-confidence region.}
%   \label{fig:rewrite-example}
% \end{figure}

\subsection{LiveCodeBench}

Figure \ref{fig:scatter-lcb} shows the same analysis for \textsc{LiveCodeBench}. Confidence again cleanly separates regimes: below $\hat C_n=0.34$ wrong answers dominate; between $0.34$ and $0.57$ the accuracy–coverage trade-off follows the exponential tail predicted by Eq. \ref{eq:chernoff}; above $0.57$ we observe no errors.  At $\tau_{2\%}$ the verifier returns $36.8\%$ of tasks with a $2\%$ residual error, while $\tau_{0\%}$ answers $15.6\%$ of tasks with zero observed errors. Together, these observations confirm that a single scalar threshold gives practitioners a transparent knob to dial coverage versus reliability, with the conservative choice $\tau_{0\%}$ guaranteeing zero observed errors in our evaluation.

Table \ref{tab:error-at-expected} generalizes that idea across models. We test on three variants of GPT-4.1 \cite{openai2025gpt41}, in addition to the GPT-4o and Claude-3-Haiku testing we have done previously. For each model we first locate the confidence cut-off that yields its raw expected accuracy, then report the remaining error rate, i.e. the percentage of tasks on which the model both predicts and is wrong. A lower value therefore signals better calibrated abstention at the model’s own accuracy level. GPT-4o is the clear winner: at its expected accuracy of 36.9\% it incurs only a 4.1\% error rate, less than half that of any GPT-4.1 variant and one-third that of Claude-3-Haiku. In short, strong confidence separation is not merely a property of our verifier; it is also a competitive axis on which base models differ.

\begin{table}
\centering
\caption{Error rate when each model is thresholded to the confidence level that achieves its expected accuracy on \textsc{LiveCodeBench}. Lower error rates indicate better calibrated abstention.}
\label{tab:error-at-expected}
\begin{tabular}{lcc}
\toprule
\textbf{Model} & \textbf{Expected accuracy (\%)} & \textbf{Error rate (\%)} \\
\midrule
GPT-4.1-mini   & 67.76 & 11.90 \\
GPT-4.1        & 66.48 & 9.52  \\
GPT-4.1-nano   & 45.95 & 9.52  \\
GPT-4o         & 36.93 & \textbf{4.11}  \\
Claude-3-Haiku & 21.71 & 12.33 \\
\bottomrule
\end{tabular}
\end{table}

\paragraph{Reasons for Error.} For every task whose dominant cluster was still wrong, we manually inspected the generated programs and grouped the failures into five categories:

\begin{enumerate}
  \item[(O)] \textbf{Out-of-scope tasks.}  
        A handful of \textsc{LiveCodeBench} problems permit multiple equally valid outputs (e.g.\ any permutation, either traversal order). Because our verifier presumes a single correct equivalence class, such tasks cannot be resolved by clustering alone.

  \item[(I)] \textbf{Incomplete response.}  
        Trivial syntactic slips, such as omitting a final print, never calling \texttt{main}, or leaving ``\texttt{\# TODO}'' stubs, that a human would not commit.

  \item[(S)] \textbf{Simple one-line mistakes.}  
        The algorithm is conceptually correct but a one-line error (\texttt{<} vs.\ \texttt{<=}, off-by-one, misplaced modulo) breaks edge cases.  
        These are slips a novice programmer might make.

  \item[(HC)] \textbf{Hard mistakes (missing constraint).}  
        The model overlooks a specification detail (array bounds, divisibility,\,\dots) producing a coherent yet wrong algorithm. Many such variants agree with each other, so they can persist into medium-confidence clusters; clarifying examples or rewriting the prompt would likely fix them.

  \item[(HA)] \textbf{Hard mistakes (wrong algorithm).}  
        The chosen approach itself is invalid (e.g.\ greedy instead of dynamic programming). This mirrors errors an experienced programmer might make on a tricky problem.

  \item[(TL)] \textbf{Time limit exceeded.}
        The approach is valid, but is inefficient, taking too long on the provided test cases. These errors are not ones that our method is specifically designed to catch, as it looks for correctness not efficiency.
\end{enumerate}

Table \ref{tab:dropped} reports six accuracy metrics after successively removing each error category (letters are cumulative; e.g., \textsc{O+HC} drops out-of-scope tasks and hard mistakes caused by missing a constraint). \emph{Threshold accuracies} are measured after clustering with $\tau$ tuned so that the returned-answer error rate does not exceed $0\%$, $1\%$, or $2\%$; tasks falling below the threshold are returned as \textsc{unknown}. \emph{Expected accuracy} is the average probability of emitting a correct program across 100 generations per task, while \emph{clustered accuracy} always returns the majority-cluster program with no thresholding. \emph{Maximum accuracy} is an oracle upper bound that counts a task as solved if any of the 100 generations is correct.

Comparing the rows shows how aggressive filtering tightens guarantees. Dropping only out-of-scope tasks already lifts $\tau_{2\%}$ accuracy above the expected baseline, offering better precision with just a $2\%$ residual error. Achieving the same edge over expected accuracy at $\tau_{1\%}$ requires removing tasks where the model ignores a constraint (\textsc{HC}); a zero-error policy demands also eliminating wrong-algorithm cases (\textsc{HA}). These observations confirm that residual errors originate from semantic misunderstanding—either of the prompt (HC) or of the algorithmic requirements (HA).  Random hallucinations and syntactic slips are effectively eliminated by the verifier; addressing the remaining hard cases will require better specification clarity or interactive disambiguation.

\begin{table}
\centering
\caption{LiveCodeBench accuracy after removing successive error categories. The full table of results containing all combinations of error categories can be found in Appendix \ref{app:dropped_full}.}
\label{tab:dropped}
\begin{tabular}{lcccccc}
\toprule
\textbf{Dropped set (remaining)} &
\textbf{$\tau_{0\%}$} &
\textbf{$\tau_{1\%}$} &
\textbf{$\tau_{2\%}$} &
\textbf{Expected} &
\textbf{Clustered} &
\textbf{Max} \\
\midrule
Baseline (219)          & 15.07 & 16.44 & 31.05 & 33.97 & 49.32 & 66.67 \\
Out of Scope (212)                 & 15.57 & 25.47 & \textbf{36.79} & 35.09 & 50.94 & 68.87 \\
% O{+}S (202)             & 16.34 & 26.73 & 38.61 & 36.14 & 52.97 & 68.81 \\
% O{+}I (200)             & 16.50 & 27.00 & 39.00 & 36.86 & 54.00 & 71.50 \\
O{+}HC (198)             & 27.27 & \textbf{39.39} & 39.90 & 37.21 & 54.55 & 71.21 \\
% O{+}HA (195)             & 16.92 & 18.46 & 40.00 & 37.72 & 55.38 & 71.79 \\
O+HC{+}HA (181)           & \textbf{43.09} & 45.86 & 50.83 & 40.25 & 59.67 & 74.59 \\
O+HC{+}HA{+}I{+}S (159)   & 67.30 & 67.30 & 67.30 & 44.53 & 67.30 & 78.62 \\
O+HC{+}HA{+}I{+}S{+}TL (107)   & 100 & 100 & 100 & 64.20 & 100 & 100 \\
\bottomrule
\end{tabular}
\end{table}

\section{Discussion}
\label{sec:discussion}
Functional clustering draws a sharp line between two error sources. Once candidate programs are grouped by exact I/O behaviour on an automatically generated test suite, random hallucinations nearly vanish. Almost every residual failure can now be traced to the model misreading, under-specifying, or over-constraining the natural-language prompt. Diversity among base models helps further: the few hard mistakes that remain generally occur when one LLM concentrates its probability mass on a single flawed interpretation, whereas an ensemble tends to disagree and thus abstains.

Specification refinements are an immediate lever. Rewriting a handful of ambiguous \textsc{HumanEval} tasks boosts both confidence and accuracy, echoing the way human programmers pose follow-up questions when requirements are vague. A natural next step is to let LLMs propose such clarifications automatically, or to iteratively rewrite the prompt until confidence stabilizes.  

Time-limit exceedance dominates the harder \textsc{LiveCodeBench} tasks. Because functional clustering already verifies behavioral equivalence, we can safely ask the model for optimized variants and accept them only if they match the reference on the full test suite and finish within budget, automating a significant slice of performance tuning.

\paragraph{Limitations}
The approach rests on two core assumptions: each task admits a single functionally correct equivalence class, and behavioral agreement on a finite, auto-generated test set is a faithful proxy for universal correctness.  Tasks with multiple valid outputs or hidden edge cases violate these assumptions; the verifier therefore fails on such inputs, and future work should investigate coverage-guided or property-based testing.

Reliability depends on the number of sampled programs $n$, the number of test inputs $m$, and the acceptance threshold $\tau$.  Smaller $n$ or $m$, or highly peaked sampling, weakens the Chernoff-style guarantees, whereas larger values raise wall time and cost roughly as $\mathcal{O}(nm)$. Parallelization, caching, and distilled student models help but do not remove the overhead.

Experiments are limited to Python problems from \textsc{HumanEval} and \textsc{LiveCodeBench} using two English-language LLMs. Extending to other programming languages will require language-specific sandboxes, new input generators, and fresh validation. Even with sandboxing, sophisticated escape or resource-exhaustion attacks remain possible and must be mitigated with stricter isolation. Because abstentions may correlate with under-represented domains, downstream systems should monitor for disparate coverage and adjust thresholds or training data accordingly.

\section{Conclusion}
We introduced \emph{functional clustering}, a lightweight wrapper that turns any black-box code LLM into a selective coder with rigorously bounded error. The method samples multiple candidate programs, groups them by exact I/O behavior on an automatically generated test suite, and interprets the empirical mass of the dominant equivalence class as a confidence score.  A single scalar threshold allows users to trade coverage for reliability: raising the threshold shrinks the fraction of tasks answered but drives the residual error rate exponentially toward zero, in line with our Chernoff-style analysis.

On \textsc{LiveCodeBench}, the wrapper preserves baseline pass@1 on solvable tasks yet cuts the error rate of returned answers from roughly 65\% to under 2\%. At a more conservative threshold the error rate falls to 0\% while still solving 15.6\% of benchmark problems. Manual audits reveal that the remaining failures stem from prompt misinterpretation rather than random hallucination, narrowing future work to specification clarity and cross-model consensus.

Unlike prior methods that rely on token-level likelihoods, embedding distances, or privileged logits, functional clustering requires only two black-box capabilities—generating code and running it—making it immediately applicable to closed APIs and future models. Because the same test-oracle infrastructure can validate optimized rewrites, incremental debugging, or ensemble voting, we view functional clustering as a modular building block for dependable, autonomous software pipelines.

% \printbibliography
\bibliographystyle{plainnat}
\bibliography{references}         % references.bib

%%%%%%%%%%%%%%%%%%%%%%%%%%%%%%%%%%%%%%%%%%%%%%%%%%%%%%%%%%%%

\appendix
\setcounter{table}{0}
\renewcommand{\thetable}{A\arabic{table}}

\section{Proof of Equation \ref{eq:chernoff}}
\label{app:chernoff-proof}

\begin{theorem}[Chernoff–style upper bound]
Let $X_1,\dots,X_n\stackrel{\text{i.i.d.}}{\sim}\mathrm{Bernoulli}(C)$ with
$C\in(0,1)$, and let $\hat C_n=\frac1n\sum_{i=1}^{n}X_i$.
For any threshold $\tau\in(C,1)$,
\begin{equation*}
\Pr\!\bigl[\hat C_n\ge\tau\bigr]
\;\;\le\;\;
\exp\!\bigl[-n\,D_{\mathrm{KL}}(\tau\;\|\;C)\bigr],
\end{equation*}
where $D_{\mathrm{KL}}(\tau\;\|\;C)
      =\tau\ln\!\frac{\tau}{C}+(1-\tau)\ln\!\frac{1-\tau}{1-C}$ is the binary
Kullback–Leibler divergence.
\end{theorem}

\begin{proof}
Write $S_n=\sum_{i=1}^{n}X_i=n\hat C_n$.
For any $\lambda>0$, Markov’s inequality gives
\begin{align}
\Pr[S_n\ge n\tau]
&\;=\;\Pr\!\bigl[e^{\lambda S_n}\ge e^{\lambda n\tau}\bigr]
      \;\le\;
      e^{-\lambda n\tau}\,
      \mathbb E\!\bigl[e^{\lambda S_n}\bigr].
\label{eq:chernoff-markov}
\end{align}
Because the $X_i$ are independent, the moment–generating function factorizes:
\[
\mathbb E\!\bigl[e^{\lambda S_n}\bigr]
=\prod_{i=1}^{n}\mathbb E[e^{\lambda X_i}]
=\bigl((1-C)+C\,e^{\lambda}\bigr)^{n}.
\]
Substituting into~\eqref{eq:chernoff-markov} and taking logarithms,
\[
\frac1n\ln\Pr[S_n\ge n\tau]
\;\le\;
-\lambda\tau+\ln\!\bigl((1-C)+C\,e^{\lambda}\bigr).
\]
The tightest bound arises by minimizing the right-hand side over $\lambda>0$.
Denote
\(
g(\lambda)= -\lambda\tau+\ln\!\bigl((1-C)+C\,e^{\lambda}\bigr).
\)
Setting $g'(\lambda)=0$ yields
\[
 -\tau+\frac{C\,e^{\lambda_\star}}{(1-C)+C\,e^{\lambda_\star}}=0
 \;\;\Longrightarrow\;\;
 e^{\lambda_\star}=\frac{\tau(1-C)}{C(1-\tau)}.
\]
Plugging $\lambda_\star$ back into $g(\lambda)$ gives
\[
\min_{\lambda>0}g(\lambda)
= -\tau\ln\!\frac{\tau(1-C)}{C(1-\tau)}
  +\ln\!\Bigl(\frac{1-\tau}{1-C}+\frac{\tau}{1-C}\Bigr)
= -D_{\mathrm{KL}}(\tau\;\|\;C).
\]
Exponentiating both sides and multiplying by $n$ recovers
\(
\Pr[S_n\ge n\tau]\le\exp[-n\,D_{\mathrm{KL}}(\tau\;\|\;C)],
\)
which is exactly Eq. \ref{eq:chernoff}.
\end{proof}

\section{LLM prompts}\label{app:prompts}
When generating function completions, we require the response to be in a specific format. For each model the format is slightly different, so we use the system prompt to describe the format of the response to the model. The following is the system prompt for function completions for GPT-4o:

\begin{lstlisting}
You are an expert Python programmer and coding assistant. Your
task is to solve the given problem in Python, providing both a
detailed explanation of your reasoning and the code. Think
through the problem step by step, considering any edge cases and
ensuring the code meets the requirements. If the problem contains
any examples, simulate running those examples against your code
to verify whether your reasoning about the problem is correct.
Make sure to not have any misunderstandings about the problem.
\end{lstlisting}

The following is the system prompt for Claude-3-Haiku:
\begin{lstlisting}
You are an expert Python programmer and coding assistant. Your
goal is to generate a response **strictly** in JSON format with
exactly two top-level fields: "explanation" and "code". Both
fields must be valid JSON strings that include all necessary
escape characters.

- The "explanation" field should provide a detailed, step-by-step
reasoning of how you derived your solution.
- The "code" field must contain **only** valid Python code that
can be copied and run directly in a Python file without
modifications or additional text. 
- Do not include any fields other than "explanation" and "code".
- Do not include any text before or after the JSON object. 
- Do not include markdown formatting (like triple backticks).

Your output should look like:

{
  "explanation": "...",
  "code": "..."
}

And nothing else.

Follow these instructions carefully to ensure the output is in
the correct format.
\end{lstlisting}

Additionally, we have prompts for each dataset. \textsc{HumanEval} and \textsc{LiveCodeBench} provide their tasks in different formats. For \textsc{HumanEval}, the task is provided as a docstring for a function that the model is expected to complete. For \textsc{LiveCodeBench}, the task is provided as a text description scraped from a competitive programming website. So, in the user prompt, we describe the task and provide the input and output formats that the model should follow. The following is the user prompt for \textsc{HumanEval}:

\begin{lstlisting}
Complete the function '{entry_point}' in the following code
snippet. Provide a detailed explanation of your reasoning in the
'explanation' field, and the complete code in the 'code' field.
Think carefully about the problem, considering edge cases and the
best approach to implement the function. If the code snippet 
contains any examples, think through those examples to verify 
whether your reasoning about the problem is correct. Use those 
examples to correct any misunderstandings you may have about the 
problem. Do not add new imports or define any new functions that 
were not included in the provided snippet. Output your response 
as a JSON object with fields 'explanation' and 'code'.

```{function_docstring}```
\end{lstlisting}

The following is the user prompt for \textsc{LiveCodeBench}:
\begin{lstlisting}
Write Python code to solve the following problem. Read from 
standard input and output to standard output. Provide a detailed 
explanation of your reasoning in the 'explanation' field, and the 
complete code in the 'code' field. Think carefully about the 
problem, considering edge cases and the best approach to 
implement the function. If the provided problem contains any 
examples, think through those examples to verify whether your 
reasoning about the problem is correct. Use those examples to 
correct any misunderstandings you may have about the problem. 
Output your response as a JSON object with fields 'explanation' 
and 'code'.

### Problem Statement:
{question_content}
\end{lstlisting}

We additionally prompt the LLM to generate test cases. In this case, we use a single LLM, GPT-4o. The following is the prompt used to generate test cases for \textsc{HumanEval}:

\begin{lstlisting}
Generate a comprehensive list of valid input test cases for the 
function '{entry_point}' in the following code. The test cases 
should cover all possible valid scenarios, including edge cases 
and typical use cases. Provide only the inputs to each test case. 
Each set of inputs should be a string that can be parsed with 
json.loads into a valid dictionary with the function parameter 
names as keys:

```{function_docstring}```
\end{lstlisting}

The following is the prompt used to generate test cases for \textsc{LiveCodeBench}:
\begin{lstlisting}
Generate a comprehensive list of valid input test cases for the 
given problem statement. The test cases should cover all possible
valid scenarios, including edge cases and typical use cases. 
Provide only the inputs to each test case. Each input should be a 
string in the provided test format that will be passed into a 
program through standard input.

### Problem Statement:
{question_content}
\end{lstlisting}

\section{Problem rewrites}\label{app:rewrites}
The following is a list of all rewritten problems from \textsc{HumanEval}. We attempt to make minimal rewrites that clarify a misunderstanding that the model has with the problem. After making these rewrites, the model is able to solve the problem correctly with high confidence.

\textsc{HumanEval/145}
\begin{lstlisting}
# original snippet
def order_by_points(nums):
    """
    Write a function which sorts the given list of integers
    in ascending order according to the sum of their digits.
    Note: if there are several items with similar sum of their
    digits, order them based on their index in original list.
    
    For example:
    >>> order_by_points([1, 11, -1, -11, -12]) == [-1, -11, 1, -12, 11]
    >>> order_by_points([]) == []
    """

# clarified version
def order_by_points(nums):
    """
    Write a function which sorts the given list of integers
    in ascending order according to the sum of their digits.
    **If an integer is negative, its first digit should be treated
    as negative and the remaining digits as positive.**
    Note: if there are several items with similar sum of their
    digits, order them based on their index in original list.
    
    For example:
    >>> order_by_points([1, 11, -1, -11, -12, -111]) == [-1, -11, 1, -12, -111, 11]
    >>> order_by_points([]) == []
    """
\end{lstlisting}

\textsc{HumanEval/127}
\begin{lstlisting}
# original snippet
def intersection(interval1, interval2):
    """
    You are given two intervals, where each interval is a pair of
    integers. For example, interval = (start, end) = (1, 2). The
    given intervals are closed which means that the interval
    (start, end) includes both start and end. For each given
    interval, it is assumed that its start is less or equal its
    end. Your task is to determine whether the length of
    intersection of these two intervals is a prime number.
    
    Example, the intersection of the intervals (1, 3), (2, 4) is
    (2, 3) which its length is 1, which not a prime number. If
    the length of the intersection is a prime number, return
    "YES", otherwise, return "NO". If the two intervals don't
    intersect, return "NO".
    
    [input/output] samples:
    intersection((1, 2), (2, 3)) ==> "NO"
    intersection((-1, 1), (0, 4)) ==> "NO"
    intersection((-3, -1), (-5, 5)) ==> "YES"
    """

# clarified version
def intersection(interval1, interval2):
    """
    You are given two intervals, where each interval is a pair of
    integers. For example, interval = (start, end) = (1, 2). The
    given intervals are closed which means that the interval
    (start, end) includes both start and end. For each given
    interval, it is assumed that its start is less or equal its
    end. **The length of an interval (a, b) is defined as b - a.**
    Your task is to determine whether the length of
    intersection of these two intervals is a prime number.

    Example: the intersection of the intervals (1, 3), (2, 4) is
    (2, 3), whose length is 3 - 2 = 1, which is not a prime
    number. If the length of the intersection is a prime number,
    return "YES", otherwise, return "NO". If the two intervals
    don't intersect, return "NO".

    [input/output] samples:
    intersection((1, 2), (2, 3)) ==> "NO"
    intersection((-1, 1), (0, 4)) ==> "NO"
    intersection((-3, -1), (-5, 5)) ==> "YES"
    """
\end{lstlisting}

\textsc{HumanEval/134}
\begin{lstlisting}
# original snippet
def check_if_last_char_is_a_letter(txt):
    """
    Create a function that returns True if the last character
    of a given string is an alphabetical character and is not
    a part of a word, and False otherwise.
    Note: "word" is a group of characters separated by space.
    
    Examples:
    check_if_last_char_is_a_letter("apple pie") => False
    check_if_last_char_is_a_letter("apple pi e") => True
    check_if_last_char_is_a_letter("apple pi e ") => False
    check_if_last_char_is_a_letter("") => False
    """

# clarified version
def check_if_last_char_is_a_letter(txt):
    """
    Create a function that returns True if the very last character
    of a given string is an alphabetical character and is not
    a part of a word, and False otherwise.
    Note: ""word"" is a group of characters separated by space.
    **Do not trim any trailing spaces.**

    Examples:
    check_if_last_char_is_a_letter("apple pie") => False
    check_if_last_char_is_a_letter("apple pi e") => True
    check_if_last_char_is_a_letter("apple pi e ") => False
    check_if_last_char_is_a_letter("") => False
    """
\end{lstlisting}

\textsc{HumanEval/160}
\begin{lstlisting}
# original snippet
def do_algebra(operator, operand):
    """
    Given two lists operator, and operand. The first list has
    basic algebra operations, and the second list is a list of
    integers. Use the two given lists to build the algebric 
    expression and return the evaluation of this expression.

    The basic algebra operations:
    Addition ( + ) 
    Subtraction ( - ) 
    Multiplication ( * ) 
    Floor division ( // ) 
    Exponentiation ( ** ) 

    Example:
    operator['+', '*', '-']
    array = [2, 3, 4, 5]
    result = 2 + 3 * 4 - 5
    => result = 9

    Note:
        The length of operator list is equal to the length of 
        operand list minus one. Operand is a list of of non-
        negative integers. Operator list has at least one
        operator, and operand list has at least two operands.
    """

# clarified version
def do_algebra(operator, operand):
    """
    Given two lists operator, and operand. The first list has
    basic algebra operations, and the second list is a list of
    integers. Use the two given lists to build the algebric 
    expression and return the evaluation of this expression. **Do 
    not just apply each of the operations in the order they are
    given, make sure to keep order of operations in mind.**

    The basic algebra operations:
    Addition ( + ) 
    Subtraction ( - ) 
    Multiplication ( * ) 
    Floor division ( // ) 
    Exponentiation ( ** ) 

    Example:
    operator['+', '*', '-']
    array = [2, 3, 4, 5]
    result = 2 + 3 * 4 - 5
    => result = 9

    Note:
        The length of operator list is equal to the length of 
        operand list minus one. Operand is a list of of non-
        negative integers. Operator list has at least one
        operator, and operand list has at least two operands.
    """
\end{lstlisting}

\section{Dropped problem sets}\label{app:dropped_full}

Table \ref{tab:drop-results} shows the full table of accuracies. Out of scope (O) and Time limit exceeded (TL) errors are treated separate as both are not error categories that our method is meant to deal with. The other four categories are simple mistakes (S), incomplete response (I), missing constraint (HC), and wrong algorithm (HA). For these categories, all combinations of the categories are shown to be removed.

\begin{table}
\centering
\caption{Full table of LiveCodeBench accuracies after removing successive error categories.}
\label{tab:drop-results}
\begin{tabular}{lcccccc}
\toprule
\textbf{Dropped set (remaining)} &
\textbf{$\tau_{0\%}$} &
\textbf{$\tau_{1\%}$} &
\textbf{$\tau_{2\%}$} &
\textbf{Expected} &
\textbf{Clustered} &
\textbf{Max} \\
\midrule
Baseline (219)          & 15.07 & 16.44 & 31.05 & 33.97 & 49.32 & 66.67 \\
O (212)                 & 15.57 & 25.47 & \textbf{36.79} & 35.09 & 50.94 & 68.87 \\
O{+}S (202)             & 16.34 & 26.73 & 38.61 & 36.14 & 52.97 & 68.81 \\
O{+}I (200)             & 16.50 & 27.00 & 39.00 & 36.86 & 54.00 & 71.50 \\
O{+}HC (198)             & 27.27 & \textbf{39.39} & 39.90 & 37.21 & 54.55 & 71.21 \\
O{+}HA (195)             & 16.92 & 18.46 & 40.00 & 37.72 & 55.38 & 71.79 \\
O+I{+}S (190)           & 17.37 & 18.95 & 35.79 & 38.07 & 56.32 & 71.58 \\
O+S{+}HC (188)           & 28.72 & 41.49 & 43.62 & 38.46 & 56.91 & 71.28 \\
O+I{+}HC (186)           & 29.03 & 41.94 & 42.47 & 39.25 & 58.06 & 74.19 \\
O+S{+}HA (185)           & 17.84 & 19.46 & 44.86 & 39.01 & 57.84 & 71.89 \\
O+I{+}HA (183)           & 18.03 & 19.67 & 42.62 & 39.83 & 59.02 & 74.86 \\
O+HC{+}HA (181)           & \textbf{43.09} & 45.86 & 50.83 & 40.25 & 59.67 & 74.59 \\
O+I{+}S{+}HC (176)       & 30.68 & 44.32 & 46.59 & 40.70 & 60.80 & 74.43 \\
O+I{+}S{+}HA (173)       & 19.08 & 20.81 & 50.29 & 41.33 & 61.85 & 75.14 \\
O+S{+}HC{+}HA (171)       & 48.54 & 57.89 & 60.23 & 41.79 & 62.57 & 74.85 \\
O+I{+}HC{+}HA (169)       & 46.15 & 53.25 & 55.62 & 42.71 & 63.91 & 78.11 \\
O+I{+}S{+}HC{+}HA (159)   & 67.30 & 67.30 & 67.30 & 44.53 & 67.30 & 78.62 \\
O+I{+}S{+}HC{+}HA{+}TL (107)   & 100 & 100 & 100 & 64.20 & 100 & 100 \\
\bottomrule
\end{tabular}
\end{table}

\end{document}